\documentclass{article}

\usepackage[final]{corl_2018} 

\usepackage[export]{adjustbox}
\title{The Lyapunov Neural Network: Adaptive Stability Certification for Safe Learning of Dynamical Systems}

\author{
    Spencer M.~Richards \\
    Department of Mechanical and Process Engineering \\
    ETH Z\"urich \\
    \texttt{spenrich@stanford.edu}
    \And
    Felix Berkenkamp \\
    Department of Computer Science \\
    ETH Z\"urich \\
    \texttt{befelix@inf.ethz.ch}
    \And
    Andreas Krause \\
    Department of Computer Science \\
    ETH Z\"urich \\
    \texttt{krausea@ethz.ch}
}

\usepackage{algorithm}                      
\usepackage{algpseudocode}                  
\usepackage{amsmath}                        
\usepackage{amssymb}                        
\usepackage{amsthm}                         
\usepackage{bm}                             
\usepackage{cancel}                         
\usepackage{subcaption}
\usepackage[capitalise,nameinlink]{cleveref}                       

\usepackage{commath}
\usepackage{mathtools}                      
\usepackage{siunitx}                        
\usepackage{units}
\usepackage[utf8]{inputenc}
\usepackage{pgfplots}
\pgfplotsset{compat=1.14}
\usepgfplotslibrary{fillbetween}
\usetikzlibrary{patterns}
\pgfdeclarelayer{bg}
\pgfsetlayers{bg,main}

\DeclareCaptionLabelSeparator{periodspace}{.\quad}
\crefname{equation}{}{} 
\crefname{section}{Sec.}{Sec.}

\newcommand{\rbr}[2][-1]{\del[#1]{#2}}

\usepackage{dsfont}

\algnewcommand\Title[1]{\item[\textbf{Algorithm:}] \textsc{#1}}
\algnewcommand\Input[1]{\State \textbf{Input:} #1}
\algnewcommand\Initialize[1]{\State \textbf{initialize} #1}
\algnewcommand\Break{\State\textbf{break}}
\algnewcommand\Continue{\State\textbf{continue}}
\algrenewcommand\Return{\State\textbf{return}~}

\newcommand{\ie}{i.e.,~}                             
\newcommand{\eg}{e.g.,~}                             

\newcommand{\x}{\times}                             
\newcommand{\defn}{\coloneqq}                      

\DeclareMathOperator{\sign}{sign}                   
\renewcommand{\vec}[1]{\boldsymbol{\mathbf{#1}}}    
\newcommand{\trans}[1]{#1^{\!\top}\!}                 
\newcommand{\R}{\mathbb{R}}                         
\newcommand{\N}{\mathbb{N}}                         

\newtheoremstyle{mystyle1}                          
    {\topsep}                                           
    {\topsep}                                           
    {\itshape}                                      
    {}                                              
    {\bfseries}                                     
    {:}                                             
    { }                                             
    {\thmname{#1}\thmnumber{ #2}\thmnote{ (#3)}}    

\theoremstyle{mystyle1}
    \newtheorem{theorem}{Theorem}
    
    \newtheorem{corollary}{Corollary}

    \newtheorem{remark}{Remark}

\usepackage{thm-restate}
\usetikzlibrary{pgfplots.groupplots}
\renewcommand\paragraph[1]{\textbf{#1}~~}

\definecolor{matblue}{rgb}{0, 0.447, 0.741}
\definecolor{matred}{rgb}{0.85, 0.325, 0.098}
\definecolor{matyellow}{rgb}{0.929, 0.694, 0.125}
\definecolor{matpurple}{rgb}{0.494, 0.184, 0.556}

\definecolor{blind0}{RGB}{0,114,178}    
\definecolor{blind1}{RGB}{213,94,0}     
\definecolor{blind2}{RGB}{0,158,115}    
\definecolor{blind3}{RGB}{240,228,66}   

\begin{document}
\maketitle


\begin{abstract}
    Learning algorithms have shown considerable prowess in simulation by allowing robots to adapt to uncertain environments and improve their performance. However, such algorithms are rarely used in practice on safety-critical systems, since the learned policy typically does not yield any safety guarantees. That is, the required exploration may cause physical harm to the robot or its environment. In this paper, we present a method to learn accurate safety certificates for nonlinear, closed-loop dynamical systems. Specifically, we construct a neural network Lyapunov function and a training algorithm that adapts it to the shape of the largest safe region in the state space. The algorithm relies only on knowledge of inputs and outputs of the dynamics, rather than on any specific model structure. We demonstrate our method by learning the safe region of attraction for a simulated inverted pendulum. Furthermore, we discuss how our method can be used in safe learning algorithms together with statistical models of dynamical systems.
\end{abstract}




\section{Introduction}
\label{sec:introduction}

Safety is among the foremost open problems in robotics and artificial intelligence~\cite{Amodei:2016}. Many autonomous systems, such as self-driving cars and robots for palliative care, are safety-critical due to their interaction with human life. At the same time, learning is necessary for these systems to perform well in \textit{a priori} unknown environments. During learning, they must \emph{safely explore} their environment by avoiding dangerous states from which they cannot recover. For example, consider an autonomous robot in an outdoor environment affected by rough terrain and adverse weather conditions. These factors introduce uncertainty about the relationship between the robot's speed and maneuverability. While the robot should learn about its capabilities in such conditions, it must not perform a maneuver at a high speed that would cause it to crash. Conversely, traveling at only slow speeds to avoid accidents is not conducive to learning about the extent of the robot's capabilities.

To ensure \emph{safe learning}, we must verify a \emph{safety certificate} for a state before it is explored. In control theory, a set of states is safe if system trajectories are bounded within it and asymptotically converge to a fixed point under a fixed control policy. Within such a \emph{region of attraction (ROA)}~\cite{Khalil:2002}, the system can collect data during learning and can always recover to a known safe point. In this paper, we leverage Lyapunov stability theory to construct provable, neural network-based safety certificates, and \emph{adapt} them to the size and shape of the largest ROA of a general nonlinear dynamical system.


\paragraph{Related work} Lyapunov functions are convenient tools for stability (\ie safety) certification of dynamical systems~\cite{Khalil:2002} and for ROA estimation~\cite{Vannelli:1985,Hill:1990,Silva:2005}. These functions encode long-term behaviour of state trajectories in a scalar value~\cite{Kalman:1960}, such that a ROA can be encoded as a level set of the Lyapunov function. However, Lyapunov functions for general dynamical systems are difficult to find; computational approaches are surveyed in~\cite{Giesl:2016}. A Lyapunov function can be identified efficiently via a semi-definite program (SDP,~\cite{Boyd:2009}) when the dynamics are polynomial and the Lyapunov function is restricted to be a sum-of-squares (SOS) polynomial~\cite{Parrilo:2000}. Other methods to compute ROAs include maximization of a measure of ROA volume over system trajectories~\cite{Henrion:2014}, and sampling-based approaches that generalize information about stability at discrete points to a continuous region~\cite{Bobiti:2016}.

This paper is particularly concerned with safety certificates for dynamical systems with uncertainties in the form of \emph{model errors}. In robust control~\cite{Zhou:1998}, the formulation of SDPs with SOS Lyapunov functions is used to compute ROA estimates for uncertain linear dynamical systems with the assumption of a worst-case linear perturbation from a known bounded set~\cite{Trofino:2000,Topcu:2010}. Learning-based control methods with a Gaussian process (GP,~\cite{Rasmussen:2006}) model of the system instead consider uncertainty in a Bayesian manner, where model errors are reduced in regions where data has been collected. The methods in~\cite{Berkenkamp:2016,Berkenkamp:2017} estimate a ROA with Lyapunov stability certificates computed on a discretization of the state space, which is used for safe reinforcement learning (RL,~\cite{Sutton:2018}). The Lyapunov function is assumed to be given in~\cite{Berkenkamp:2016}, while~\cite{Berkenkamp:2017} uses the negative value (\ie cost) function from RL with a quadratic reward. Ultimately, this approach is limited by a \emph{shape mismatch} between level sets of the Lyapunov function and the true largest ROA. For example, a quadratic Lyapunov function has ellipsoidal level sets, which cannot characterize a non-ellipsoidal ROA, while the SOS approach is restricted to fixed monomial features. To improve safe exploration for general nonlinear dynamics, we want to \emph{learn} these features to determine a Lyapunov function with suitably shaped level sets.


\paragraph{Contributions} In this paper, we present a novel method for learning accurate safety certificates for general nonlinear dynamical systems. We construct a neural network Lyapunov candidate and, unlike past work in~\cite{Petridis:2006,Noroozi:2008}, we structure our candidate such that it \emph{always} inherently yields a provable safety certificate. Then, we specify a training algorithm that adapts the candidate to the shape of the dynamical system's trajectories via classification of states as safe or unsafe. We do not depend on any specific structure of the dynamics for this. We show how our construction relates to SOS Lyapunov functions, and compare our approach to others on a simulated inverted pendulum benchmark. We also discuss how our method can be used to make safe learning more effective.


\section{Problem Statement and Background}
\label{sec:Background}

We consider a discrete-time, time-invariant, deterministic dynamical system of the form
\begin{equation}
    \vec{x}_{t + 1} = f(\vec{x}_t, \vec{u}_t),
\end{equation}
where $t \in \N$ is the time step index, and $\vec{x}_t \in \mathcal{X} \subset \R^d$ and $\vec{u}_t \in \mathcal{U} \subset \R^p$ are the state and control inputs respectively at time step~$t$. The system is controlled by a feedback policy $\pi \colon \mathcal{X} \to \mathcal{U}$ and the resulting closed-loop dynamical system is given by $\vec{x}_{t + 1} = f_\pi(\vec{x}_t)$ with $f_\pi(\vec{x}) = f(\vec{x}, \pi(\vec{x}))$. We assume this policy is given, but it can, for example, be computed online with RL or optimal control. This policy~$\pi$ is safe to use within a subset~$\mathcal{S}_\pi$ of the state space~$\mathcal{X}$. The set~$\mathcal{S}_\pi$ is a ROA for~$f_\pi$, \ie every system trajectory of $f_\pi$ that begins at some $\vec{x} \in \mathcal{S}_\pi$ also remains in $\mathcal{S}_\pi$ and asymptotically approaches an \emph{equilibrium point} $\vec{x}_\mathrm{O} \in \mathcal{S}_\pi$ where $f_\pi(\vec{x}_\mathrm{O}) = \vec{x}_\mathrm{O}$~\cite{Khalil:2002}. We assume $\vec{x}_\mathrm{O} = \vec{0}$ without loss of generality. Hereafter, we use $\mathcal{S}_\pi$ to denote the true largest ROA in $\mathcal{X}$ under the policy~$\pi$.

A reliable estimate of~$\mathcal{S}_\pi$ is critical to online learning systems, since we need to ensure that a policy is safe to use on the real system before it can be deployed. The goal of this paper is to estimate the largest safe set~$\mathcal{S}_\pi$. We must also ensure safety by never overestimating~$\mathcal{S}_\pi$, \ie we must not identify unsafe states as safe. For this to be feasible, we make a regularity assumption about the closed-loop dynamics; we assume $f_\pi$ is Lipschitz continuous on $\mathcal{X}$ with Lipschitz constant $L_{f_\pi} \in \R_{> 0}$. This is a weak assumption and is even satisfied when a neural network policy is used~\cite{Szegedy:2014}.


\subsection{Safety Certification with Lyapunov Functions}
\label{sec:SafetyCertification}

One way to estimate the safe region~$\mathcal{S}_\pi$ is by using a Lyapunov function. Given a suitable Lyapunov function $v$, a safe region for the closed-loop dynamical system $\vec{x}_{t+1} = f_\pi(\vec{x}_t)$ can be determined.

\begin{theorem}[Lyapunov's stability theorem {\normalfont \cite{Kalman:1960}}]
    \label{thm:Lyapunov}
    Suppose $f_\pi$ is locally Lipschitz continuous and has an equilibrium point at $\vec{x}_\mathrm{O} = \vec{0}$. Let $v : \mathcal{X} \to \R$ be locally Lipschitz continuous on $\mathcal{X}$. If there exists a set $\mathcal{D}_v \subseteq \mathcal{X}$ containing $\vec{0}$ on which $v$ is positive-definite and $\Delta v(\vec{x}) \defn v(f_\pi(\vec{x})) - v(\vec{x}) < 0$, $\forall \vec{x} \in \mathcal{D}_v \setminus \{\vec{0}\}$, then $\vec{x}_\mathrm{O} = \vec{0}$ is an asymptotically stable equilibrium. In this case, $v$ is known as a Lyapunov function for the closed-loop dynamics $f_\pi$, and $\mathcal{D}_v$ is the Lyapunov decrease region for $v$.
\end{theorem}

\cref{thm:Lyapunov} states that a Lyapunov function~$v$ characterizes a ``basin'' of safe states where trajectories of~$f_\pi$ ``fall'' towards the origin $\vec{x}_\mathrm{O} = \vec{0}$. If we can find a positive-definite $v$ such that the dynamics always map downwards in the value of $v(\vec{x})$, then trajectories eventually reach $v(\vec{x}) = \vec{0}$, thus $\vec{x} = \vec{0}$. To find a ROA, rather than checking if $v$ decreases along entire trajectories, it is sufficient to verify the \emph{one-step decrease condition} $\Delta v(\vec{x}) < 0$ for every state $\vec{x}$ in a level set of $v$.

\begin{corollary}[Safe level sets {\normalfont \cite{Kalman:1960}}]
    \label{thm:LevelSets}
    Every level set $\mathcal{V}(c) \defn \cbr{\vec{x} \mid v(\vec{x}) \leq c}, c \in \R_{> 0}$ contained within the decrease region $\mathcal{D}_v$ is invariant under $f_\pi$. That is, $f_\pi(\vec{x}) \in \mathcal{V}(c), \forall \vec{x} \in \mathcal{V}(c)$. Furthermore, $\lim_{t \to \infty} \vec{x}_t = \vec{0}$ for every $\vec{x}_t$ in these level sets, so each one is a ROA for $f_\pi$ and $\vec{x}_\mathrm{O} = \vec{0}$.
\end{corollary}

Intuitively, if $v(\vec{x})$ decreases everywhere in the level set $\mathcal{V}(c_1)$, except at $\vec{x}_{\mathrm{O}} = \vec{0}$ where it is zero, then $\mathcal{V}(c_1)$ is invariant, since the image of $\mathcal{V}(c_1)$ under $f_\pi$ is the smaller level set $\mathcal{V}(c_2)$ with $c_2 < c_1$. If $v$ is also positive-definite, then this ensures trajectories that start in a level set $\mathcal{V}(c)$ contained in the decrease region $\mathcal{D}_v$ remain in $\mathcal{V}(c)$ and converge to $\vec{x}_\mathrm{O} = \vec{0}$. To identify safe level sets, we must check if a given \emph{Lyapunov candidate} $v$ satisfies the conditions of \cref{thm:Lyapunov}. However, the decrease condition $\Delta v(\vec{x}) < 0$ is difficult to verify throughout a continuous subset $\mathcal{D}_v \subseteq \mathcal{X}$. It is sufficient to verify the tightened safety certificate~$\Delta v(\vec{x}) < -L_{\Delta v} \tau$ at a finite set of points that cover~$\mathcal{D}_v$, where~$L_{\Delta v} \in \R_{> 0}$ is the Lipschitz constant of~$\Delta v$ and~$\tau \in \R_{>0}$ is a measure of how densely the points cover~$\mathcal{D}_v$~\cite{Berkenkamp:2017}. We can even couple this with bounds on $f_\pi$ from a statistical model to certify high-probability safe sets with the certificate $\Delta \hat{v}(\vec{x}) < -L_{\Delta v} \tau$, where~$\Delta \hat{v}(\vec{x})$ is an upper confidence bound on~$\Delta v(\vec{x})$. A GP model of $f_\pi$ is used for this purpose in~\cite{Berkenkamp:2017}.


\subsection{Computing SOS Lyapunov Functions}
\label{sec:SOSLyapunovFunctions}

In general, a suitable Lyapunov candidate $v$ is difficult to find. Computational methods often restrict $v$ to a particular function class for tractability. The SOS approach restricts $v(\vec{x})$ to be polynomial, but is limited to polynomial dynamical systems, \ie when $f_\pi(\vec{x})$ is a vector of polynomials in the elements of $\vec{x}$ \cite{Parrilo:2000,Papachristodoulou:2002,Papachristodoulou:2005}. In particular, the SOS approach enforces $v(\vec{x}) = \trans{m(\vec{x})}\vec{Q}m(\vec{x})$, where~$m(\vec{x})$ is a vector of \emph{a priori} fixed monomial features in the elements of $\vec{x}$, and $\vec{Q}$ is an unknown positive-semidefinite matrix. This makes~$v(\vec{x})$ a quadratic function on a monomial \emph{feature space}. A SDP can be efficiently solved to yield a $\vec{Q}$ that \emph{simultaneously} guarantees that $v$ satisfies the assumptions of \cref{thm:Lyapunov} and has the largest possible level set in its decrease region $\mathcal{D}_v$. That is, the positive-definiteness of $v$ and the negative-definiteness of $\Delta v$ in $\mathcal{D}_v$ are enforced as constraints in the SDP. This contrasts the more general approach described in \cref{sec:SafetyCertification}, where a Lyapunov candidate $v$ is given and then the assumptions of \cref{thm:Lyapunov} are verified by checking discrete points.

\begin{figure}[t]
    \captionsetup[subfigure]{justification=centering}
    \centering

\pgfmathsetmacro{\rays}{10}
\pgfmathsetmacro{\radius}{0.5cm}
\pgfmathsetmacro{\clipwidth}{1.58}
\pgfmathsetmacro{\clipheight}{1.60}
\pgfmathsetmacro{\scale}{1.21}

\begin{subfigure}[c]{0.5\textwidth}
    \centering
    \begin{tikzpicture}[scale=\scale]
        \clip (-\clipwidth,-\clipheight) rectangle (\clipwidth,\clipheight);

        \draw[fill] (0,0) circle (0.5pt) node[below] {$\vec{x}_\mathrm{O} = \vec{0}$};

        \node[] at (-0.8cm, 0.4cm) {$\mathcal{S}_\pi$};
        \pgfmathsetseed{8}
        \draw[very thick, smooth cycle, tension=1, samples=8, domain={0:\rays}, shift={(0.04cm, 0.08cm)}]
            plot (\x*360/\rays:\radius+1cm*rnd);

        \node[] at (0.25cm, 0.45cm) {\color{blind0} $\mathcal{V}(c)$};
        \pgfmathsetmacro{\length}{1.13}
        \pgfmathsetmacro{\width}{0.56}
        \pgfmathsetmacro{\rot}{55}
        \draw[thick, blind0, rotate=\rot, fill=blind0, fill opacity=0.3] (0,0) ellipse (\length cm and \width cm);

        \node[] at (0.25, 1.27) {\color{blind2} $\mathcal{D}_v$};
        \pgfmathsetmacro{\curv}{5}
        \pgfmathsetmacro{\ext}{1.1}
        \pgfmathsetmacro{\ang}{20}
        \pgfmathsetmacro{\dist}{15}
        \begin{scope}[rotate=\rot]
            \draw[dashed, very thick, smooth, domain=-\ext:\ext, blind2] plot (\x, \curv*\x*\x + \width cm)
                to[out=\ang,in=-\ang,distance=\dist] (\ext,-\curv*\ext*\ext - \width cm);
            \draw[dashed, very thick, smooth, domain=\ext:-\ext, blind2] plot (\x, -\curv*\x*\x - \width cm)
                to[out=180+\ang,in=180-\ang,distance=\dist] (-\ext, \curv*\ext*\ext + \width cm);
        \end{scope}
    \end{tikzpicture}
    \vspace{-7pt}
    \caption{Shape mismatch with a fixed\\Lyapunov function.\label{fig:ShapeMismatch}}
\end{subfigure}%
\begin{subfigure}[c]{0.5\textwidth}
    \centering
    \begin{tikzpicture}[scale=\scale]
        \clip (-\clipwidth,-\clipheight) rectangle (\clipwidth,\clipheight);

        \draw[fill] (0,0) circle (0.5pt) node[below] {$\vec{x}_\mathrm{O} = \vec{0}$};

        \pgfmathsetseed{8}
        \node[] at (-1.1cm, -1.2cm) {\color{blind1} $y = -1$};
        \node[] at (1.15cm, -1.2cm) {\color{blind1} $y = -1$};
        \node[] at (-0.9cm, 1cm) {\color{blind1} $y = -1$};

        \node[] at (-0.2cm, 0.4cm) {$\mathcal{S}_\pi = {\color{blind0} \mathcal{V}_{\vec{\theta}}(c_\mathcal{S})}$};
        \pgfmathsetseed{8}
        \draw[very thick, smooth cycle, tension=0.8, samples=8, domain={0:\rays}, shift={(0.04cm, 0.08cm)}]
            plot (\x*360/\rays:\radius+1cm*rnd);

        \node[] at (0.25cm, -0.75cm) {\color{blind0} $y = +1$};
        \pgfmathsetseed{8}
        \draw[very thick, smooth cycle, tension=0.8, samples=8, domain={0:\rays}, shift={(0.04cm, 0.08cm)}, fill=blind0, fill opacity=0.3]
            plot (\x*360/\rays:\radius+1cm*rnd);
    \end{tikzpicture}
    \vspace{-7pt}
    \caption{Shape match with a parameterized\\Lyapunov function.\label{fig:SafeClassifier}}
\end{subfigure}%
    \caption{\cref{fig:ShapeMismatch} illustrates a shape mismatch between the largest level set $\mathcal{V}(c)$ (blue ellipsoid) of a quadratic Lyapunov function $v$ contained within the decrease region $\mathcal{D}_v$ (green dashes), and the safe region $\mathcal{S}_\pi$ (black). We cannot certify all of~$\mathcal{S}_\pi$ with $v$, which limits exploration in safe learning. Instead, we train a Lyapunov candidate $v_{\vec{\theta}}$ with parameters $\vec{\theta}$ to match~$\mathcal{S}_\pi$ with a level set $\mathcal{V}_{\vec{\theta}}(c_\mathcal{S})$, as in \cref{fig:SafeClassifier}, via classification of sampled states as ``safe'' with ground-truth label $y = +1$ (\ie $\vec{x} \in \mathcal{S}_\pi$) or ``unsafe'' with $y = -1$ (\ie $\vec{x} \notin \mathcal{S}_\pi$).}
\end{figure}
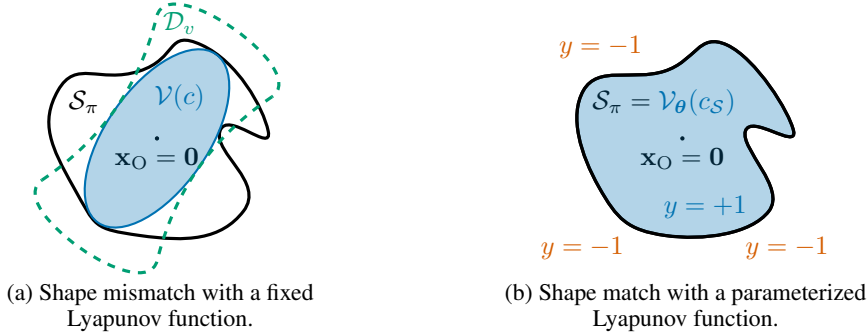

With the SOS approach and a suitable choice of $m(\vec{x})$, $\mathcal{S}_\pi$ can be estimated well with a level set $\mathcal{V}(c)$ of~$v$, since the monomial features allow Lyapunov functions with shapes beyond simple ellipsoids to be found. However, the SOS approach requires polynomial dynamics, and the best choice of $m(\vec{x})$ can be difficult to determine. Without a suitable Lyapunov function, we face the problem of a \emph{shape mismatch} between $\mathcal{V}(c)$ and $\mathcal{S}_\pi$. This is exemplified in~\cref{fig:ShapeMismatch}, where level sets of quadratic $v$ are ellipsoidal while $\mathcal{S}_\pi$ is not, which limits the region of the state space that is certifiable as safe by $v$.


\section{Learning Lyapunov Candidates}
\label{sec:Theory}

In this section, we establish a more flexible class of parameterized Lyapunov candidates that can satisfy the assumptions on $v$ in \cref{thm:Lyapunov} by virtue of their structure and gradient-based parameter training. In particular, we show how a binary classification problem based on whether each state $\vec{x}$ lies within the safe region $\mathcal{S}_\pi$ can be formulated to train the parameterized Lyapunov candidate.

\subsection{Construction of a Neural Network Lyapunov Function}

We take the SOS approach in \cref{sec:SOSLyapunovFunctions} as a starting point to construct a neural network Lyapunov candidate. The SOS Lyapunov candidate $v(\vec{x}) = \trans{m(\vec{x})}\vec{Q}m(\vec{x})$ is a Euclidean inner product on the transformed space $\mathcal{Y} \defn \cbr[1]{\phi(\vec{x}),\ \forall \vec{x} \in \mathcal{X}}$ with $\phi(\vec{x}) \defn \vec{Q}^{1/2}m(\vec{x})$. The ability of the SOS Lyapunov candidate $v$ to certify safe states for~$f_\pi$ depends on the choice of monomials in~$m(\vec{x})$. We interpret these choices as engineered features that define the expressiveness of~$v$ in delineating the decision boundary between safe and unsafe states. Rather than choose such features manually and parameterize~$\phi(\vec{x})$ with~$\vec{Q}$ only, we propose the Lyapunov candidate~$v_{\vec{\theta}}(\vec{x}) = \trans{\phi_{\vec{\theta}}(\vec{x})}\phi_{\vec{\theta}}(\vec{x})$ to \emph{learn} the requisite features, where~$\phi_{\vec{\theta}} : \R^d \to \R^D$ is a feed-forward neural network with parameter vector~$\vec{\theta}$. Feed-forward neural networks are expressive in that they can approximate any continuous function on compact subsets of $\R^d$ with a finite number of parameters \cite{Cybenko:1989,Hornik:2001}. In \cref{sec:LearningSafeSet}, we exploit this property together with gradient-based parameter training to closely match the true ROA~$\mathcal{S}_\pi$ with a level set of the candidate $v_{\vec{\theta}}$ without the need to engineer individual features of $\phi$.

We cannot use an arbitrary feed-forward neural network~$\phi_{\vec{\theta}}$ in our Lyapunov candidate, since the conditions of \cref{thm:Lyapunov} must be satisfied. Otherwise, the resulting candidate~$v_{\vec{\theta}}$ cannot provide any safety information. In general, $\phi_{\vec{\theta}}$ is a sequence of function compositions or layers. Each layer has the form $\vec{y}_\ell(\vec{x}) = \varphi_\ell(\vec{W}_\ell\vec{y}_{\ell-1}(\vec{x}))$, where~$\vec{y}_\ell(\vec{x})$ is the output of layer~$\ell$ for state~$\vec{x} \in \mathcal{X}$, $\varphi_\ell$ is a fixed element-wise activation function, and $\vec{W}_\ell\vec{y}_{\ell-1}(\vec{x})$ is a linear transformation parameterized by $\vec{W}_\ell \in \R^{d_\ell \x d_{\ell-1}}$. To satisfy the assumptions of \cref{thm:Lyapunov}, $v_{\vec{\theta}}$ must be Lipschitz continuous on $\mathcal{X}$ and positive-definite on some subset of $\mathcal{X}$ around $\vec{x}_\mathrm{O} = \vec{0}$. To this end, we restrict $v_{\vec{\theta}}$ to be positive-definite and Lipschitz continuous on $\mathcal{X}$ for \emph{all} values of~$\vec{\theta} \defn \cbr{\vec{W}_\ell}_\ell$ with a suitable choice of structure for $\phi_{\vec{\theta}}$.

\begin{restatable}[Lyapunov neural network]{theorem}{LyapunovNeuralNetwork}
    \label{thm:NeuralNetwork}
    Consider $v_{\vec{\theta}}(\vec{x}) = \trans{\phi_{\vec{\theta}}(\vec{x})}\phi_{\vec{\theta}}(\vec{x})$ as a Lyapunov candidate function, where $\phi_{\vec{\theta}}$ is a feed-forward neural network. Suppose, for each layer~$\ell$ in~$\phi_{\vec{\theta}}$, the activation function $\varphi_\ell$ and weight matrix $\vec{W}_\ell \in \R^{d_\ell \x d_{\ell-1}}$ each have a trivial nullspace. Then~$\phi_{\vec{\theta}}$ has a trivial nullspace, and $v_{\vec{\theta}}$ is positive-definite with $v_{\vec{\theta}}(\vec{0}) = 0$ and $v_{\vec{\theta}}(\vec{x}) > 0,\ \forall \vec{x} \in \mathcal{X} \setminus \cbr{\vec{0}}$. Furthermore, if~$\varphi_\ell$ is Lipschitz continuous for each layer~$\ell$, then~$v_{\vec{\theta}}$ is locally Lipschitz continuous.
\end{restatable}

We provide a formal proof of \cref{thm:NeuralNetwork} in \cref{app:Proofs} and briefly outline it here. As an inner product, $v_{\vec{\theta}}(\vec{x}) = \trans{\phi_{\vec{\theta}}(\vec{x})}\phi_{\vec{\theta}}(\vec{x})$ is already positive-definite for any neural network output~$\phi_{\vec{\theta}}(\vec{x})$, and thus is \emph{at least} nonnegative for any state~$\vec{x} \in \mathcal{X}$. The step from nonnegativity to positive-definiteness of~$v_{\vec{\theta}}$ on~$\mathcal{X}$ now only depends on how the origin~$\vec{0} \in \mathcal{X}$ is mapped through~$\phi_{\vec{\theta}}$. If~$\phi_{\vec{\theta}}$ maps $\vec{0} \in \mathcal{X}$ \emph{uniquely} to the zero output $\phi_{\vec{\theta}}(\vec{0}) = \vec{0}$, \ie if~$\phi_{\vec{\theta}}$ has a trivial nullspace, then~$v_{\vec{\theta}}$ is positive-definite. For this, it is sufficient that each layer of~$\phi_{\vec{\theta}}$ has a trivial nullspace, \ie that each layer ``passes along'' $\vec{0} \in \mathcal{X}$ to its zero output $\vec{y}_\ell(\vec{0}) = \vec{0}$ until the final output~$\phi_{\vec{\theta}}(\vec{0}) = \vec{0}$.

In \cref{thm:NeuralNetwork}, each layer~$\ell$ has a trivial nullspace as long as its weight matrix~$\vec{W}_\ell$ and activation function~$\varphi_\ell$ have trivial nullspaces. Consequently, this requires that $d_\ell \geq d_{\ell-1}$ for each layer~$\ell$, where~$d_\ell$ is the output dimension of layer~$\ell$. That is, $\vec{W}_\ell$ must not decrease the dimension of its input. To ensure that~$\vec{W}_\ell$ has a trivial nullspace, we structure it as
\begin{equation}
    \label{eqn:LyapunovMatrix}
    \vec{W}_\ell = \begin{bmatrix} \trans{\vec{G}_{\ell 1}}\vec{G}_{\ell 1} + \varepsilon\vec{I}_{d_{\ell-1}} \\ \vec{G}_{\ell 2}\end{bmatrix},
\end{equation}
where~$\vec{G}_{\ell 1} \in \R^{q_\ell \x d_{\ell-1}}$ for some $q_\ell \in \mathbb{N}_{\geq 1}$, ${\vec{G}_{\ell 2} \in \R^{(d_\ell - d_{\ell-1}) \x d_{\ell-1}}}$, $\vec{I}_{d_{\ell-1}} \in \R^{d_{\ell-1} \x d_{\ell-1}}$ is the identity matrix, and $\varepsilon \in \R_{>0}$ is a constant. The top partition $\trans{\vec{G}_{\ell 1}}\vec{G}_{\ell 1} + \varepsilon\vec{I}_{d_{\ell-1}}$ is positive-definite for~$\varepsilon > 0$, thus~$\vec{W}_\ell$ always has full rank and a trivial nullspace. Otherwise, $\vec{W}_\ell$ would have a non-empty nullspace of dimension~$d_{\ell-1} - \min(d_\ell, d_{\ell-1}) = d_{\ell-1} - d_\ell > 0$ by the rank-nullity theorem. With this choice of structure for~$\vec{W}_\ell$, the parameters of the neural network~$\phi_{\vec{\theta}}$ are given by~$\vec{\theta} \defn \cbr{\vec{G}_{\ell 1}, \vec{G}_{\ell 2}}_\ell$. Finally, we choose activation functions that have trivial nullspaces and that are Lipschitz continuous in $\mathcal{X}$, such as $\tanh(\cdot)$ and the leaky ReLU. We can then compute a Lipschitz constant for $\phi_{\vec{\theta}}$ \cite{Szegedy:2014}.

\subsection{Learning a Safe Set via Classification}
\label{sec:LearningSafeSet}

Previously, we constructed a neural network Lyapunov candidate~$v_{\vec{\theta}}$ in~\cref{thm:NeuralNetwork} that satisfies the positive-definiteness and Lipschitz continuity requirements in \cref{thm:Lyapunov}. As a result, we can always use the one-step decrease condition~$\Delta v_{\vec{\theta}}(\vec{x}) \defn v_{\vec{\theta}}(f_\pi(\vec{x})) - v_{\vec{\theta}}(\vec{x}) < 0$ as a provable safety certificate to identify safe level sets that are subsets of the largest safe region~$\mathcal{S}_\pi$. Now, we design a training algorithm to adapt the parameters~$\vec{\theta}$ such that the resulting Lyapunov candidate~$v_{\vec{\theta}}$ satisfies~$\Delta v_{\vec{\theta}}(\vec{x}) < 0$ throughout as large of a decrease region~$\mathcal{D}_{v_{\vec{\theta}}} \subseteq \mathcal{X}$ as possible. This also makes~$v_{\vec{\theta}}$ a valid Lyapunov function for the closed-loop dynamics~$f_\pi$.

For now, we assume the entire safe region $\mathcal{S}_\pi$ is known. We want to use a level set~$\mathcal{V}_{\vec{\theta}}(c)$ of $v_{\vec{\theta}}$ to certify the entire set $\mathcal{S}_\pi$ as safe. According to \cref{thm:Lyapunov}, this requires the Lyapunov decrease condition~$\Delta v_{\vec{{\theta}}}(\vec{x}) < 0$ to be satisfied for each state $\vec{x} \in \mathcal{S}_\pi$. We formally state this problem as
\begin{equation}
    \label{eqn:ROAVol}
    \max_{\vec{\theta}, c} \operatorname{\mathrm{Vol}\!}\rbr{\mathcal{V}_{\vec{\theta}}(c) \cap \mathcal{S}_\pi},~\mathrm{s.t.}~\Delta v_{\vec{\theta}}(\vec{x}) < 0, \forall \vec{x} \in \mathcal{V}_{\vec{\theta}}(c),
\end{equation}
where $\operatorname{\mathrm{Vol}}(\cdot)$ is some measure of set volume. Thus, we want to find the largest level set of~$v_{\vec{\theta}}$ that is contained in the true largest ROA $\mathcal{S}_\pi$; see~\cref{fig:roa_initial}. We fix~$c = c_\mathcal{S}$ with some~$c_\mathcal{S} \in \R_{> 0}$, as it is always possible to rescale~$v_{\vec{\theta}}$ by a constant, and focus on optimizing over $\vec{\theta}$. We can then interpret~\cref{eqn:ROAVol} as a classification problem. Consider~\cref{fig:SafeClassifier}, where we assign the ground-truth label~$y = +1$ whenever a state~$\vec{x}$ is contained in~$\mathcal{S}_\pi$, and~$y = -1$ otherwise. We use~$v_{\vec{\theta}}$ together with~\cref{thm:Lyapunov} to classify states by their membership in the level set $\mathcal{V}(c_\mathcal{S})$. This is described by the decision rule
\begin{equation}
    \label{eqn:decision_rule}
    \hat{y}_{\vec{\theta}}(\vec{x}) = \sign\rbr[1]{c_\mathcal{S} - v_{\vec{\theta}}(\vec{x})}.
\end{equation}
That is, each state within the level set $\mathcal{V}(c_\mathcal{S})$ obtains the label~$y=+1$. However, we must also satisfy the Lyapunov decrease condition imposed by~\cref{thm:Lyapunov}. This can be written as the constraint
\begin{equation}
    \label{eqn:ROALogicB}
    y = +1 \implies \Delta v_{\vec{\theta}}(\vec{x}) < 0,
\end{equation}
which means that we can assign the label~$y=+1$ only if the decrease condition is also satisfied. The decision rule~\cref{eqn:decision_rule} together with the constraint~\cref{eqn:ROALogicB} ensures that the resulting estimated safe set~$\mathcal{V}(c_\mathcal{S})$ satisfies all of the conditions in~\cref{thm:Lyapunov}. We want to select the neural network parameters $\vec{\theta}$ so that this rule can perfectly classify $\vec{x} \in \mathcal{S}_\pi$ as ``safe'' with $\hat{y}_{\vec{\theta}}(\vec{x}) = +1$ (\ie ${c_\mathcal{S} - v_{\vec{\theta}}(\vec{x}) > 0}$) or $\vec{x} \notin \mathcal{S}_\pi$ as ``unsafe'' with $\hat{y}_{\vec{\theta}}(\vec{x}) = -1$ (\ie $c_\mathcal{S} - v_{\vec{\theta}}(\vec{x}) \leq 0$). To this end, the decision boundary $v_{\vec{\theta}}(\vec{x}) = c_\mathcal{S}$ must exactly delineate the boundary of $\mathcal{S}_\pi$. Furthermore, the value of $\vec{\theta}$ must ensure \cref{eqn:ROALogicB} holds, such that $v_{\vec{\theta}}$ satisfies the decrease condition of \cref{thm:Lyapunov} on $\mathcal{S}_\pi$.

\begin{figure}[t]
    \centering

\pgfmathsetmacro{\length}{1.135}
\pgfmathsetmacro{\width}{0.5}
\pgfmathsetmacro{\rot}{60}
\pgfmathsetmacro{\gain}{1.6}
\pgfmathsetmacro{\rays}{10}
\pgfmathsetmacro{\radius}{0.5cm}

\pgfmathsetmacro{\clipwidth}{1.53}
\pgfmathsetmacro{\clipheight}{1.64}
\pgfmathsetmacro{\scale}{1.35}

\begin{subfigure}[c]{0.33\textwidth}
    \centering
    \begin{tikzpicture}[scale=\scale]
        \clip (-\clipwidth,-\clipheight) rectangle (\clipwidth,\clipheight);

        \draw[fill] (0,0) circle (0.5pt) node[below] {$\vec{x}_\mathrm{O} = \vec{0}$};

        \node[] at (0.25cm, 0.45cm) {\color{blind0} $\mathcal{V}_{\vec{\theta}}(c_{k})$};
        \draw[thick, blind0, rotate=\rot, fill=blind0, fill opacity=0.3] (0,0) ellipse (\length cm and \width cm);

        \node[] at (-0.9cm, 0.35cm) {$\mathcal{S}_\pi$};
        \pgfmathsetseed{8}
        \draw[very thick, smooth cycle, tension=1, samples=8, domain={0:\rays}, shift={(0.04cm, 0.08cm)}]
            plot (\x*360/\rays:\radius+1cm*rnd);
    \end{tikzpicture}
    \caption{Current safe level set.\label{fig:roa_initial}}
\end{subfigure}%
\begin{subfigure}[c]{0.33\textwidth}
    \centering
    \begin{tikzpicture}[scale=\scale]
        \clip (-\clipwidth,-\clipheight) rectangle (\clipwidth,\clipheight);

        \draw[fill] (0,0) circle (0.5pt) node[below] {$\vec{x}_\mathrm{O} = \vec{0}$};

        \node[] at (0.25cm, 0.45cm) {\color{blind0} $\mathcal{V}_{\vec{\theta}}(c_{k})$};
        \node[] at (0.6cm, 1.25cm) {\color{blind1} $\mathcal{V}_{\vec{\theta}}(\alpha c_{k})$};

        \draw[thick, blind1, rotate=\rot, fill=blind1, fill opacity=0.3, even odd rule]
            (0,0) ellipse (\length cm and \width cm)
            (0,0) ellipse (\gain * \length cm and \gain * \width cm);
        \path[rotate=\rot, name path=A]
            (0,0) ellipse (\gain * \length cm and \gain * \width cm);
        \draw[thick, blind0, rotate=\rot, fill=blind0, fill opacity=0.3]
            (0,0) ellipse (\length cm and \width cm);

        \path[rotate=\rot, name path=A]
            (0,0) ellipse (\gain * \length cm and \gain * \width cm);
        \path[rotate=\rot, name path=B]
            (0,0) ellipse (\length cm and \width cm);

        \begin{pgfonlayer}{bg}
            \node[] at (-0.95cm, 0.35cm) {$\mathcal{S}_\pi$};
            \pgfmathsetseed{8}
            \path[clip, smooth cycle, tension=1, samples=8, domain={0:\rays}, shift={(0.04cm, 0.08cm)},
                  preaction={draw, very thick}]
                plot (\x*360/\rays:\radius+1cm*rnd);
            \tikzfillbetween[of=A and B, on layer=bg] {blind2, fill, fill opacity=1};
        \end{pgfonlayer}
        \draw [-stealth, very thick]
            (0.4cm, -1.3cm) node[right] {\color{blind2!75!blind1} $\mathcal{G} \cap \mathcal{S}_\pi$}
            -- (0.2cm, -0.75cm);
    \end{tikzpicture}
    \caption{Simulate gap states forward. \label{fig:roa_expansion}}
\end{subfigure}%
\begin{subfigure}[c]{0.33\textwidth}
    \centering
    \begin{tikzpicture}[scale=\scale]
        \clip (-\clipwidth,-\clipheight) rectangle (\clipwidth,\clipheight);

        \draw[fill] (0,0) circle (0.5pt) node[below] {$\vec{x}_\mathrm{O} = \vec{0}$};

        \node[] at (-0.95cm, 0.35cm) {$\mathcal{S}_\pi$};
        \pgfmathsetseed{8}
        \path[clip, smooth cycle, tension=1, samples=8, domain={0:\rays}, shift={(0.04cm, 0.08cm)},
              preaction={draw, very thick}]
            plot (\x*360/\rays:\radius+1cm*rnd);

        \node[] at (0.25cm, 0.45cm) {\color{blind0} $\mathcal{V}_{\vec{\theta}}(c_{k+1})$};
        \path[clip, rotate=\rot, preaction={draw, thick, blind0, fill=blind0, fill opacity=0.3}]
            (0,0) ellipse (\gain * \length cm and \gain * \width cm);
        \pgfmathsetseed{8}
        \draw[blind0, very thick, smooth cycle, tension=1, samples=8, domain={0:\rays}, shift={(0.04cm, 0.08cm)}]
            plot (\x*360/\rays:\radius+1cm*rnd);
    \end{tikzpicture}
    \caption{Re-shape safe level set. \label{fig:roa_final}}
\end{subfigure}
    \caption{Illustration of training the parameterized Lyapunov candidate $v_{\vec{\theta}}$ to expand the safe level set $\mathcal{V}_{\vec{\theta}}(c_k)$ (blue ellipsoid) towards the true largest ROA $\mathcal{S}_\pi$ (black). States in the gap $\mathcal{G}$ between $\mathcal{V}_{\vec{\theta}}(c_k)$ and $\mathcal{V}_{\vec{\theta}}(\alpha c_k)$ (orange ellipsoid) are simulated forward to determine regions (green) towards which we can expand the safe level set. This information is used in \cref{alg:ROAClassifier} to iteratively adapt safe level sets of $v_{\vec{\theta}}$ to the shape of $\mathcal{S}_\pi$.}
    \label{fig:AlgorithmVisualized}
\end{figure}
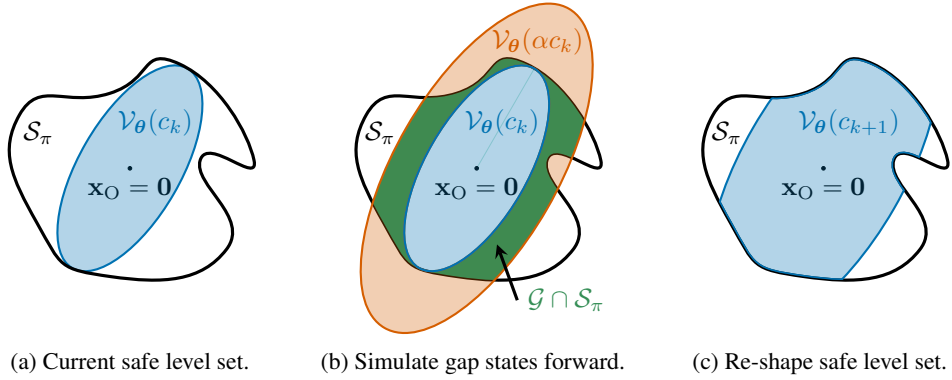

Since we have rewritten the optimization problem in~\cref{eqn:ROAVol} as a classification problem, we can use ideas from the corresponding literature~\cite{Bishop:2006}. In particular, we define a loss function~$\ell(y, \vec{x}; \vec{\theta})$ that penalizes misclassification of the true label~$y$ at a state~$\vec{x}$ under the decision rule~\cref{eqn:decision_rule} associated with~$\vec{\theta}$. Many common choices for the loss function are possible; for simplicity, we use the perceptron loss, which penalizes misclassifications more when they occur far from the decision boundary. We choose not to use the ``maximum margin'' objective of the hinge loss, since it may be unsuitable for us to accurately delineate~$\mathcal{S}_\pi$, where states can lie arbitrarily close to the decision boundary in the continuous state space~$\mathcal{X}$. Since we use the level set~$\mathcal{V}_{\vec{\theta}}(c_\mathcal{S})$ in our classification setting, this corresponds to~$\ell(y, \vec{x}; \vec{\theta}) = \max\rbr[1]{0, -y \cdot \rbr{c_\mathcal{S} - v_{\vec{\theta}}(\vec{x})}}$. Here, $c_{\mathcal{S}} - v_{\vec{\theta}}(\vec{x})$ is the signed distance from the decision boundary $v_{\vec{\theta}}(\vec{x}) = c_\mathcal{S}$, which separates the safe set $\mathcal{S}_\pi$ from the rest of the state space~$\mathcal{X} \setminus \mathcal{S}_\pi$. This \emph{classifier loss} has a magnitude of $\abs{c_\mathcal{S} - v_{\vec{\theta}}(\vec{x})}$ in the case of a misclassification, and zero otherwise. This ensures that decisions far from the decision boundary, such as those near the origin, are considered more important than the more difficult decisions close to the boundary.

Ideally, we would like to minimize this loss throughout the state space with $\min \int_\mathcal{X} l(y, \vec{x}; \theta)\!~\mathrm{d}\vec{x}$ subject to the constraint~\cref{eqn:ROALogicB}. Since this problem is intractable, we use gradient-based optimization together with mini-batches instead, as is typically done in machine learning. To this end, we sample states~$\mathcal{X}_b = \{\vec{x}_i \}_i$ from the state space~$\mathcal{X}$ at random and assign the ground-truth labels~$\{y_i\}_i$ to them. Based on this finite set, the optimization objective can be written as
\begin{equation}
    \min_{\vec{\theta}}  \sum_{\vec{x} \in \mathcal{X}_b} \ell(y, \vec{x}; \vec{\theta}),~\mathrm{s.t.}~y = +1 \implies \Delta v_{\vec{\theta}}(\vec{x}) < 0,
\end{equation}
where the batch~$\mathcal{X}_b$ is re-sampled after every gradient step. We can apply a Lagrangian relaxation
\begin{equation}
    \label{eqn:RelaxedClassifier}
    \min_{\vec{\theta}} \sum_{\vec{x} \in \mathcal{X}_b} \ell(y, \vec{x}; \vec{\theta})
    + \lambda\rbr{\frac{y + 1}{2}} \max \left( 0,\, \Delta v_{\vec{\theta}}(\vec{x}) \right)
\end{equation}
in order to make the problem tractable. Here, $\lambda \in \R_{> 0}$ is a Lagrangian multiplier and the term $\lambda((y+1)/2)\max\rbr[1]{0, \Delta v_{\vec{\theta}}(\vec{x})}$ is the \emph{Lyapunov decrease loss}, which penalizes violations of \cref{eqn:ROALogicB}. The decrease condition $\Delta v_{\vec{\theta}}(\vec{x}) < 0$ only needs to be enforced within the safe region $\mathcal{S}_\pi$, so we do not want to incur a loss if it is violated at a state where $y = -1$. Thus, we use the multiplier $(y + 1)/2$ to map $\cbr{+1, -1}$ to $\cbr{1, 0}$, such that the Lyapunov decrease loss is zeroed-out if $y = -1$.

However, there are two issues when this formulation is compared to the exact problem in~\cref{eqn:ROAVol}. Firstly, the objective~\cref{eqn:RelaxedClassifier} only penalizes violations of the decrease condition~$\Delta v_{\vec{\theta}}(\vec{x}) < 0$, rather than constraining~$\vec{\theta}$ to enforce it. Thus, while~$\Delta v_{\vec{\theta}}(\vec{x}) < 0$ is \emph{always} a provable safety certificate, we must \emph{verify} that it holds over some level set whenever we update~$\vec{\theta}$. Secondly, ground-truth labels of~$\mathcal{S}_\pi$ are not known in practice. To address these issues, we can use any method to check Lyapunov safety certificates over continuous state spaces to certify a level set $\mathcal{V}_{\vec{\theta}}(c)$ as safe, and then use $\mathcal{V}_{\vec{\theta}}(c)$ to estimate labels $y$ from $\mathcal{S}_\pi$. For this work, we check the tightened certificate~$\Delta v_{\vec{\theta}}(\vec{x}) < -L_{\Delta v_{\vec{\theta}}} \tau$ on a discretization of $\mathcal{X}$, as described in \cref{sec:SafetyCertification}. This method exposes the Lipschitz constant~$L_{\Delta v_{\vec{\theta}}}$ of~$\Delta v_{\vec{\theta}}$, which can conveniently be used for regularization in practice~\cite{Szegedy:2014}. Possible alternatives to this safety verification method include the use of an adaptive discretization for better scaling to higher-dimensional state spaces~\cite{Bobiti:2016}, and formal verification methods for neural networks~\cite{Huang:2017,Katz:2017}.

\begin{algorithm}[t]
    \caption{ROA Classifier Training} \label{alg:ROAClassifier}
    \begin{algorithmic}[1]
        \Input{closed-loop dynamics~$f_\pi$; initialized parametric Lyapunov candidate~$v_{\vec{\theta}} : \mathcal{X} \to \R_{\geq 0}$; Lagrange multiplier~$\lambda \in \R_{> 0}$; level set ``expansion'' multiplier~$\alpha \in \R_{> 1}$; forward-simulation horizon~$T \in \mathbb{N}_{\geq 1}$.}
        \State $c_0 \gets \max_{\vec{x} \in \mathcal{X}} v_{\vec{\theta}}(\vec{x}),~\mathrm{s.t.}~\mathcal{V}_{\vec{\theta}}(c_0) \subseteq \mathcal{D}_{v_{\vec{\theta}}}$.
            \Comment{\parbox[t]{0.5\linewidth}{compute the initial safe level set (\eg use a discretization, as described in \cref{sec:SafetyCertification})}}
        \Repeat
            \State Sample a finite batch $\mathcal{X}_b \subset \mathcal{V}_{\vec{\theta}}(\alpha c_k)$.
            \State $\mathcal{S}_b \gets \cbr[1]{\vec{x} \in \mathcal{X}_b \mid f^{(T)}_\pi(\vec{x}) \in \mathcal{V}_{\vec{\theta}}(c_k)}$.
                \Comment{\parbox[t]{0.5\linewidth}{forward-simulate the batch with $f_\pi$ over $T$ steps}}
            \State Update $\vec{\theta}$ with \cref{eqn:RelaxedClassifier} via batch SGD on $\mathcal{X}_b$ and labels $\cbr{y_i}_i$ for points in $\mathcal{S}_b$.
            \State $c_{k+1} \gets \max_{\vec{x} \in \mathcal{X}} v_{\vec{\theta}}(\vec{x}),~\mathrm{s.t.}~\mathcal{V}_{\vec{\theta}}(c_{k+1}) \subseteq \mathcal{D}_{v_{\vec{\theta}}}$.
        \Until{convergence}
    \end{algorithmic}
\end{algorithm}

Since such an estimate of~$\mathcal{S}_\pi$ is limited by the largest safe level set of $v_{\vec{\theta}}$, we propose \cref{alg:ROAClassifier} to iteratively ``grow'' an estimate of~$\mathcal{S}_\pi$. We initialize~$v_{\vec{\theta}}$, then use it to identify the largest safe level set~$\mathcal{V}_{\vec{\theta}}(c_0)$ by verifying the condition~$\Delta v_{\vec{\theta}}(\vec{x}) < 0$. At first, we use~$\mathcal{V}_{\vec{\theta}}(c_0)$ to estimate~$\mathcal{S}_\pi$. At iteration~$k \in \mathbb{N}_{\geq 0}$, we consider the safe level set~$\mathcal{V}_{\vec{\theta}}(c_k)$ and the expanded level set~$\mathcal{V}_{\vec{\theta}}(\alpha c_k)$ for some~$\alpha \in \R_{> 1}$; see~\cref{fig:roa_expansion}. Then, states in the ``gap''~$\mathcal{G} \defn \mathcal{V}_{\vec{\theta}}(\alpha c_k) \setminus \mathcal{V}_{\vec{\theta}}(c_k)$ are forward-simulated with the dynamics~$f_\pi$ for $T \in \mathbb{N}_{\geq 1}$ time steps. States that fall in $\mathcal{V}_{\vec{\theta}}(c_k)$ before or after forward-simulation form a new estimate of~$\mathcal{S}_\pi$, since trajectories become ``trapped'' in $\mathcal{V}_{\vec{\theta}}(c_k)$ and converge to the origin. We use this estimate of~$\mathcal{S}_\pi$ to identify labels $y$ for classification, then apply SGD with the objective \cref{eqn:RelaxedClassifier} to update $\vec{\theta}$ and encourage $\mathcal{V}_{\vec{\theta}}(c_k)$ to grow. Finally, we certify the new largest safe level set~$\mathcal{V}_{\vec{\theta}}(c_{k+1})$. These steps are repeated until a choice of stopping criterion is satisfied.

In general, \cref{alg:ROAClassifier} does not guarantee convergence of the safe level set $\mathcal{V}_{\vec{\theta}}(c_k)$ to $\mathcal{S}_\pi$, nor that~$\mathcal{V}_{\vec{\theta}}(c_k)$ monotonically grows in volume. Furthermore, it is not guaranteed that the iterated safe level~$c_k \in \R_{> 0}$ approaches the safe level~$c_{\mathcal{S}}$ that is prescribed to delineate~$\mathcal{S}_\pi$. This is typical of gradient-based parameter training, since the parameters $\vec{\theta}$ can become ``stuck'' in local optima. However, since the Lyapunov candidate $v_{\vec{\theta}}$ is guaranteed to satisfy the positive-definiteness and Lipschitz continuity conditions of \cref{thm:Lyapunov} by its construction in \cref{thm:NeuralNetwork}, $\Delta v_{\vec{\theta}}(\vec{x}) < 0$ is \emph{always} a provable safety certificate for identifying safe level sets. Thus, we can \emph{always} use~$v_{\vec{\theta}}$ to identify at least a subset of~$\mathcal{S}_\pi$, without ever identifying unsafe states as safe.


\section{Experiments and Discussion}
\label{sec:Results}

In the previous section, we developed \cref{alg:ROAClassifier} to train the parameters $\vec{\theta}$ of a neural network Lyapunov candidate $v_{\vec{\theta}}$ constructed according to \cref{thm:NeuralNetwork}. This construction ensures the positive-definiteness and Lipschitz continuity assumptions on $v_{\vec{\theta}}$ in \cref{thm:Lyapunov} are satisfied. \cref{alg:ROAClassifier} encourages $v_{\vec{\theta}}$ to satisfy the decrease condition and match the true largest ROA $\mathcal{S}_\pi$ for the closed-loop dynamics~$f_\pi$ with a level set $\mathcal{V}_{\vec{\theta}}(c_{\mathcal{S}})$. In this section, we present details for the implementation of \cref{alg:ROAClassifier} to learn the largest safe region of a simulated inverted pendulum system, and experimental results in a comparison to other methods of computing Lyapunov functions.


\paragraph{Inverted Pendulum Benchmark} The inverted pendulum is governed by the differential equation $m\ell^2\ddot{\theta} = mg\ell\sin\theta - \beta \dot{\theta} + u$ with state $\vec{x} \defn (\theta,\dot{\theta})$, where $\theta$ is the angle from the upright equilibrium~$\vec{x}_{\mathrm{O}} = \vec{0}$, $u$ is the input torque, $m$ is the pendulum mass, $g$ is the gravitational acceleration, $\ell$ is the pole length, and $\beta$ is the friction coefficient. We discretize the dynamics with a time step of~$\Delta t = \unit[0.01]{s}$ and enforce a saturation constraint $u \in [-\bar{u},\bar{u}]$, such that the pendulum falls over past a certain angle and cannot recover. For a linear policy~$u = \pi(\vec{x}) = \vec{K}\vec{x}$, this yields the safe region $\mathcal{S}_\pi$ in \cref{fig:Pendulum} around the upright equilibrium for the closed-loop dynamics $f_\pi$. In particular, we fix $\vec{K}$ to the linear quadratic regulator (LQR) solution for the discretized, linearized, unconstrained form of the dynamics~\cite{Lewis:2012}. Outside of $\mathcal{S}_\pi$, the pendulum falls down without the ability to recover and the system trajectories diverge away from~$\vec{x}_{\mathrm{O}} = \vec{0}$.

\begin{figure}[t]
    \centering
    \begin{subfigure}[b]{0.5\textwidth}
        \centering


\newenvironment{customlegend}[1][]{%
    \begingroup
    \csname pgfplots@init@cleared@structures\endcsname
    \pgfplotsset{#1}%
}{%
    \csname pgfplots@createlegend\endcsname
    \endgroup
}%

\def\addlegendimage{\csname pgfplots@addlegendimage\endcsname}

\definecolor{blind0}{RGB}{0,158,115}    
\definecolor{blind1}{RGB}{230,159,0}    
\definecolor{blind2}{RGB}{0,114,178}    
\definecolor{blind3}{RGB}{240,228,66}   

\begin{tikzpicture}
    \begin{axis}[width=5cm,
                 height=4.5cm,
                 scale only axis,
                 enlargelimits=false,
                 axis on top,
                 xlabel={angle [deg]},
                 ylabel={angular velocity [deg/s]}
    ]
        \addplot graphics[xmin=-180, xmax=180, ymin=-360, ymax=360] {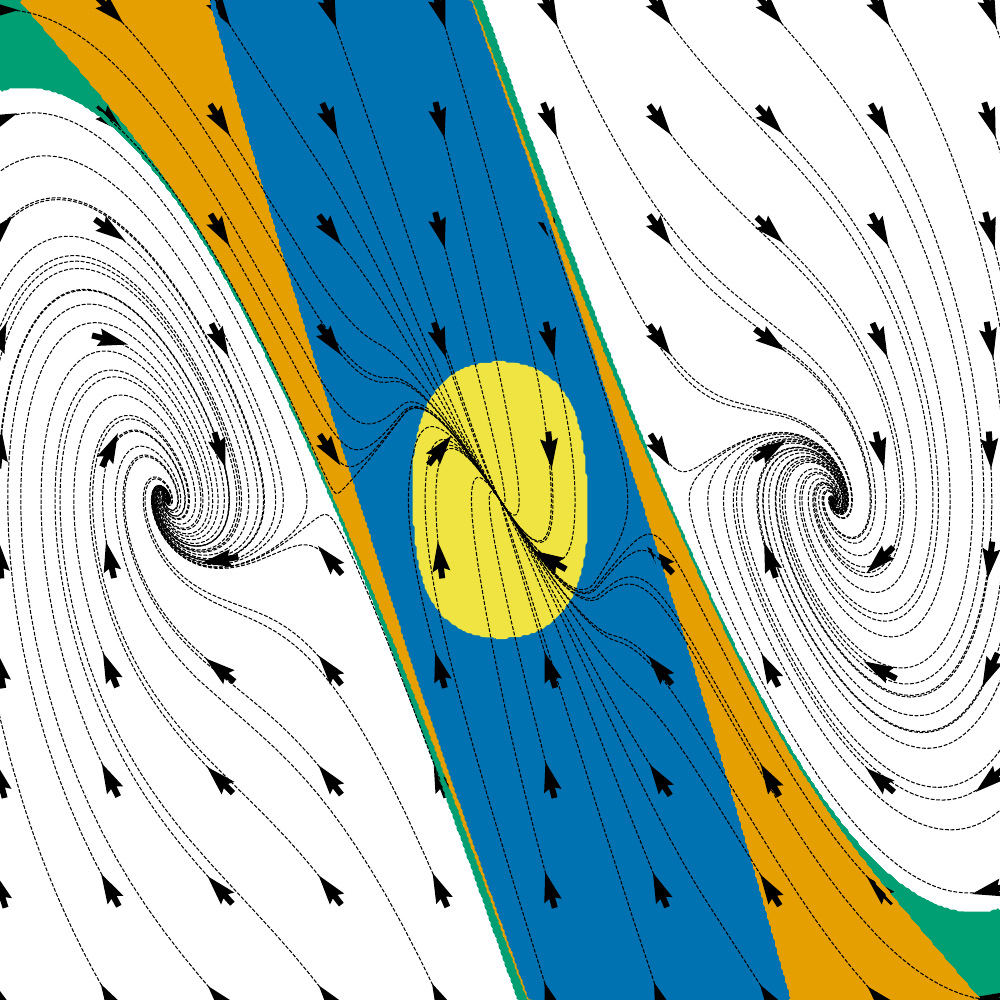};
    \end{axis}

    \begin{customlegend}[legend entries={$\mathcal{S}_\pi$, NN, LQR, SOS},
                         legend style={at={(4.95, 3.55)}, anchor=east, font=\small}]
        \addlegendimage{blind0!70!black, fill=blind0, area legend}
        \addlegendimage{blind1!70!black, fill=blind1, area legend}
        \addlegendimage{blind2!70!black, fill=blind2, area legend}
        \addlegendimage{blind3!70!black, fill=blind3, area legend}
    \end{customlegend}
\end{tikzpicture}%
        \caption{Safe Lyapunov candidate level sets.}
        \label{fig:PendulumA}
    \end{subfigure}%
    \begin{subfigure}[b]{0.5\textwidth}
        \centering

\definecolor{blind0}{RGB}{0,158,115}    
\definecolor{blind1}{RGB}{230,159,0}    
\definecolor{blind2}{RGB}{0,114,178}    
\definecolor{blind3}{RGB}{240,228,66}   

\begin{tikzpicture}
    \begin{groupplot}[group style={group size=1 by 2,
                                   xlabels at=edge bottom,
                                   xticklabels at=edge bottom,
                                   vertical sep=15pt},
                      width=5cm,
                      height=1.8cm,
                      scale only axis,
                      xmin=-0.5,
                      xmax=18.5,
                      xlabel={safe level set update iteration $k$}
    ]
        \nextgroupplot[ymin=0, ymax=1.5, ylabel={safe level $c_k$}]
        \addplot[blind1, mark=*] table[x=iter, y=c_max, col sep=comma] {figures/pendulum_data.csv};

        \nextgroupplot[ymin=0, ymax=1, ylabel={fraction of $\mathcal{S}_\pi$}]
        \addplot[blind1, mark=*] table[x=iter, y=roa_frac, col sep=comma] {figures/pendulum_data.csv};
    \end{groupplot}
\end{tikzpicture}
        \caption{Training behaviour of neural network candidate.}
        \label{fig:PendulumB}
    \end{subfigure}
    \caption{Results for training the neural network (NN) Lyapunov candidate $v_{\vec{\theta}}$ for an inverted pendulum. In \cref{fig:PendulumA}, system trajectories (black) converge to the origin only within the largest safe region $\mathcal{S}_\pi$ (green). The NN candidate (orange) characterizes $\mathcal{S}_\pi$ with a level set better than both the LQR (blue ellipsoid) and SOS (yellow) candidates, as it adapts to the shape of $\mathcal{S}_\pi$. In \cref{fig:PendulumB}, the safe level $c_k$ of $v_{\vec{\theta}}$ converges non-monotonically towards the fixed boundary $c_{\mathcal{S}} = 1$, and the safe level set $\mathcal{V}_{\vec{\theta}}(c_k)$ grows to cover most of $\mathcal{S}_\pi$. However, as discussed at the end of \cref{sec:Theory}, convergence of $\mathcal{V}_{\vec{\theta}}(c_k)$ to $\mathcal{S}_\pi$ is not guaranteed in general by \cref{alg:ROAClassifier}.}
    \label{fig:Pendulum}
\end{figure}
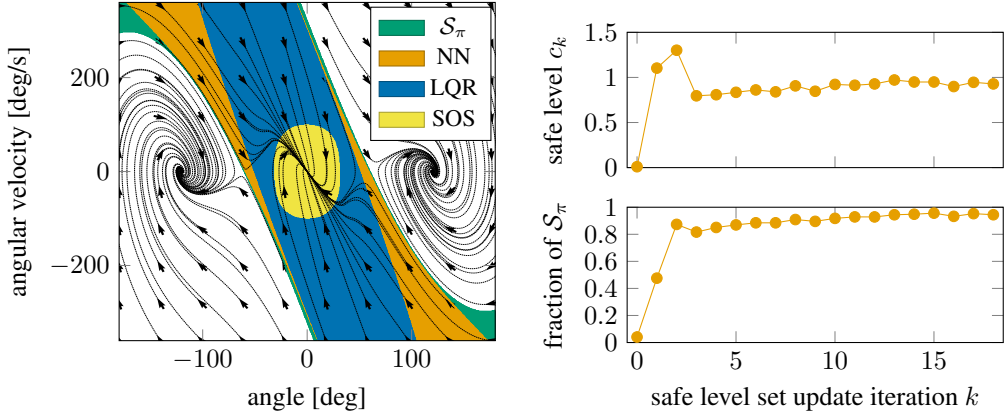


\paragraph{Practical Considerations} To train the parameters of the Lyapunov candidate $v_{\vec{\theta}}$ to adapt to the shape of $\mathcal{S}_\pi$, we use \cref{alg:ROAClassifier} with SGD. To certify the safety of continuous level sets of~$v_{\vec{\theta}}$ whenever~$\vec{\theta}$ is updated, we check the stricter decrease condition~$\Delta v_{\vec{\theta}}(\vec{x}) < -L_{\Delta v_{\vec{\theta}}}\tau$ at a discrete set of points that cover $\mathcal{X}$ in increasing order of the value of $v_{\vec{\theta}}(\vec{x})$, as in \cite{Berkenkamp:2017}. \cref{alg:ROAClassifier} does not guarantee that the safe level set estimate $\mathcal{V}_{\vec{\theta}}(c_k)$ grows monotonically in volume towards $\mathcal{S}_\pi$ with each iteration $k$. In fact, the estimate $\mathcal{V}_{\vec{\theta}}(c_k)$ may shrink if $v_{\vec{\theta}}$ initially succeeds and then fails to satisfy the decrease condition $\Delta v_{\vec{\theta}}(\vec{x}) < 0$ in some regions of the state space. This tends to occur near the origin, where~${v_{\vec{\theta}}(\vec{0}) = \Delta v_{\vec{\theta}}(\vec{0}) = 0}$ and the ``basin of attraction'' characterized by $v_{\vec{\theta}}$ ``flattens''. To alleviate this, we use a large Lagrange multiplier $\lambda = 1000$ in the SGD objective \cref{eqn:RelaxedClassifier} to strongly ``push'' $\vec{\theta}$ towards values that ensure $v_{\vec{\theta}}$ continues to satisfy the decrease condition. In addition, we normalize the Lyapunov decrease loss $\lambda((y+1)/2)\max\rbr[1]{0, \Delta v_{\vec{\theta}}(\vec{x})}$ in \cref{eqn:RelaxedClassifier} by $v_{\vec{\theta}}(\vec{x})$. This more heavily weighs sampled states near the origin, \ie where $v_{\vec{\theta}}(\vec{x})$ is small.


\paragraph{Results} We implement \cref{alg:ROAClassifier} on the inverted pendulum benchmark with the Python code available at \url{https://github.com/befelix/safe_learning}, which is based on TensorFlow~\cite{Abadi:2016}. For the neural network Lyapunov candidate $v_{\vec{\theta}}$, we use three layers of 64 $\tanh(\cdot)$ activation units each. We prescribe $\mathcal{V}_{\vec{\theta}}(c_{\mathcal{S}})$ with $c_{\mathcal{S}} = 1$ as the level set that delineates the safe region~$\mathcal{S}_\pi$. \cref{fig:Pendulum} shows the results of training $v_{\vec{\theta}}$ with \cref{alg:ROAClassifier}, and the largest safe level set $\mathcal{V}_{\vec{\theta}}(c_{18})$ with $10$ SGD iterations per update. \cref{fig:PendulumA} visualizes how this level set has ``moulded'' to the shape of~$\mathcal{S}_\pi$. \cref{fig:PendulumB} shows how the safe level~$c_k$ converges towards the prescribed level~$c_{\mathcal{S}} = 1$ that delineates~$\mathcal{S}_\pi$, and how the fraction of~$\mathcal{S}_\pi$ covered by~$\mathcal{V}_{\vec{\theta}}(c_k)$ approaches~$1$. The true largest ROA~$\mathcal{S}_\pi$ is estimated by forward-simulating all of the states in a state space discretization, and set volume is estimated by counting discrete states. \cref{fig:PendulumA} also shows the largest safe sets for a LQR Lyapunov candidate and a SOS Lyapunov candidate. The LQR candidate $v_{\mathrm{LQR}}(\vec{x}) = \trans{\vec{x}}\vec{P}\vec{x}$ is computed in closed-form for the same discretized, linearized, unconstrained form of the dynamics used to determine the LQR policy~$\pi(\vec{x}) = \vec{K}\vec{x}$ \cite{Lewis:2012}. The SOS Lyapunov candidate~$v_{\mathrm{SOS}}(\vec{x}) = \trans{m(\vec{x})}\vec{Q}m(\vec{x})$ uses up to third-order monomials in $\vec{x}$, thus it is a sixth-order polynomial. It is computed with the toolbox SOSTOOLS \cite{Prajna:2002} and the SDP solver SeDuMi \cite{Sturm:1999} in MATLAB for the unconstrained nonlinear dynamics with a Taylor polynomial expansion of $\sin\theta$. While the SOS approach is a powerful specialized method for polynomial dynamical systems, it cannot account for the non-differentiable nonlinearity introduced by the input saturation, which drastically alters the closed-loop dynamics. As a result, while $v_{\mathrm{SOS}}$ is optimized for the system without saturation, it is ill-suited to the true closed-loop dynamics and yields a small safe level set. Overall, our neural network Lyapunov candidate $v_{\vec{\theta}}$ performs the best at certification of as much of $\mathcal{S}_\pi$ as possible, since it only relies on inputs and outputs of $f_\pi$, and adapts to the shape of~$\mathcal{S}_\pi$.


\paragraph{Comments on Safe Learning} \cref{fig:PendulumA} demonstrates that a neural network Lyapunov candidate~$v_{\vec{\theta}}$ can certify more of the true largest safe region $\mathcal{S}_\pi$ than other common Lyapunov candidates. This has important implications for safe exploration during learning for dynamical systems; with more safe states available to visit, an agent can better learn about itself and its environment under a wider range of operating conditions. For example, our method is applicable in the safe reinforcement learning framework of~\cite{Berkenkamp:2017}. This past work provides safe exploration guarantees for a GP model of the dynamics~$f_\pi$ with confidence bounds on the Lyapunov stability certificate, but these guarantees are limited by the choice of Lyapunov function. As our results have shown, certain Lyapunov candidates may poorly characterize the shape of the true largest safe region~$\mathcal{S}_\pi$. Since our neural network Lyapunov candidate can adapt to the shape of~$\mathcal{S}_\pi$ during learning by using, for example, the mean estimate of $f_\pi$ from the GP model, we could enlarge the estimated safe region more quickly as data is collected. Our method is also applicable to exploration algorithms within safe motion planning that depends on knowledge of a safe region, such as in~\cite{Koller:2018}. Overall, our method strongly warrants consideration for use in safe learning methods that leverage statistical models of dynamical systems.


\section{Conclusion}
\label{sec:Conclusion}

We have demonstrated a novel method for learning safety certificates for general nonlinear dynamical systems. Specifically, we developed a flexible class of parameterized Lyapunov candidate functions and a training algorithm to adapt them to the shape of the largest safe region for a closed-loop dynamical system. We believe that our method is appealing due to its applicability to a wide range of dynamical systems in theory and practice. Furthermore, it can play an important role in improving safe exploration during learning for real autonomous systems in uncertain environments.



\clearpage


\acknowledgments{This research was supported in part by SNSF grant {200020\_159557}, the Vector Institute, and a fellowship by the Open Philanthropy Project.}


\bibliography{references}  

\clearpage

\appendix

\section{Proofs}
\label{app:Proofs}

\LyapunovNeuralNetwork*

\begin{proof}
    We begin by showing that $\phi_{\vec{\theta}}$ has a trivial nullspace in $\mathcal{X}$ by induction, and then use this to prove that $v_{\vec{\theta}}$ is positive-definite on $\mathcal{X}$. Recall that a feed-forward neural network is a successive composition of its layer transformations, such that the output $\vec{y}_\ell(\vec{x})$ of layer~$\ell$ for the state $\vec{x} \in \mathcal{X}$ is the input to layer~$\ell+1$. Consider~$\ell = 0$ with the input $\vec{y}_0(\vec{x}) \defn \vec{x}$, and the first layer output $\vec{y}_1(\vec{x}) = \varphi_1 (\vec{W}_1 \vec{y}_0(\vec{x}))$. Clearly $\vec{y}_0$ has a trivial nullspace in $\mathcal{X}$, since it is just the identity function. Since $\vec{W}_1$, $\varphi_1$, and $\vec{y}_0$ each have a trivial nullspace in their respective input spaces, the sequence of logical statements
    \begin{equation}
        \vec{x} = \vec{0} \iff \vec{y}_0(\vec{x}) = \vec{0} \iff \vec{W}_1 \vec{y}_0(\vec{x}) = \vec{0} \iff \varphi_1 (\vec{W}_1 \vec{y}_0(\vec{x})) = \vec{0}
    \end{equation}
    holds. Thus, $\vec{x} = \vec{0} \iff \varphi_1 (\vec{W}_1 \vec{y}_0(\vec{x})) = \vec{0}$ holds, and $\vec{y}_1$ has a trivial nullspace in $\mathcal{X}$. If we now assume $\vec{y}_\ell$ has a trivial nullspace in $\mathcal{X}$, it is clear that $\vec{y}_{\ell+1}$ has a trivial nullspace in $\mathcal{X}$, since
    \begin{equation}
        \vec{x} = \vec{0} \iff \vec{y}_\ell(\vec{x}) = \vec{0} \iff \vec{W}_{\ell+1} \vec{y}_\ell(\vec{x}) = \vec{0} \iff \varphi_{\ell+1} (\vec{W}_{\ell+1} \vec{y}_\ell(\vec{x})) = \vec{0}
    \end{equation}
    holds in a similar fashion. As a result, $\vec{y}_\ell$ has a trivial nullspace for each layer~$\ell$ by induction. Since $\phi_{\vec{\theta}}$ is a composition of a finite number of layers, $\phi_{\vec{\theta}} = \vec{y}_L$ for some $L \in \mathbb{N}_{\geq 0}$, thus $\phi_{\vec{\theta}}$ has a trivial nullspace in $\mathcal{X}$.

    We now use this property of $\phi_{\vec{\theta}}$ to prove that the Lyapunov candidate $v_{\vec{\theta}}(\vec{x}) = \trans{\phi_{\vec{\theta}}(\vec{x})}\phi_{\vec{\theta}}(\vec{x})$ is positive-definite on $\mathcal{X}$. As an inner product, $\trans{\phi_{\vec{\theta}}(\vec{x})}\phi_{\vec{\theta}}(\vec{x})$ is positive-definite on the transformed space $\mathcal{Y} \defn \cbr{\phi_{\vec{\theta}}(\vec{x}),\ \forall \vec{x} \in \mathcal{X}}$. Thus, $v_{\vec{\theta}}(\vec{x}) = 0 \iff \phi_{\vec{\theta}}(\vec{x}) = \vec{0}$ and $v_{\vec{\theta}}(\vec{x}) > 0$ otherwise. Since we have already proven $\phi_{\vec{\theta}}(\vec{x}) = \vec{0} \iff \vec{x} = \vec{0}$, combining these statements shows that $v_{\vec{\theta}}(\vec{x}) = 0 \iff \vec{x} = \vec{0}$ and $v_{\vec{\theta}}(\vec{x}) > 0$ otherwise. As a result, $v_{\vec{\theta}}(\vec{x})$ is positive-definite on $\mathcal{X}$.

    Finally, we need to show that if every activation function $\varphi_\ell$ is Lipschitz continuous, then~$v_{\vec{\theta}}$ is locally Lipschitz continuous. If the neural network $\phi_{\vec{\theta}}$ is Lipschitz continuous, then clearly $v_{\vec{\theta}}$ is locally Lipschitz continuous, since it is quadratic and thus differentiable with respect to $\phi_{\vec{\theta}}$. To show that $\phi_{\vec{\theta}}$ is Lipschitz continuous, it is sufficient to show that each layer is Lipschitz continuous. This is due to the fact that any function composition $f(g(\vec{x}))$ is Lipschitz continuous with Lipschitz constant $L_f L_g$ if $f$ has Lipschitz constant $L_f$ and $g$ has Lipschitz constant $L_g$. This fact can be seen from $\norm{f(g(\vec{x})) - f(g(\vec{x}'))} \leq L_f \norm{g(\vec{x}) - g(\vec{x}')} \leq L_fL_g\norm{\vec{x} - \vec{x}'}$, for each pair ${\vec{x}, \vec{x}' \in \mathcal{X}}$. By the Lipschitz continuity of function composition and the linearity of $\vec{W}_\ell \vec{y}_{\ell-1}$, each layer transformation $\vec{y}_\ell = \varphi_\ell(\vec{W}_\ell \vec{y}_{\ell-1})$ is Lipschitz continuous if $\varphi_\ell$ is Lipschitz continuous. As a result, the neural network $\phi_{\vec{\theta}}$ is Lipschitz continuous, and the Lyapunov candidate $v_{\vec{\theta}}$ is locally Lipschitz continuous.
\end{proof}

\begin{remark}
    In \cref{eqn:LyapunovMatrix}, we ensured each weight matrix $\vec{W}_\ell$ has a trivial nullspace with the structure
    \begin{equation*}
        \vec{W}_\ell = \begin{bmatrix} \trans{\vec{G}_{\ell 1}}\vec{G}_{\ell 1} + \varepsilon\vec{I}_{d_{\ell-1}} \\ \vec{G}_{\ell 2}\end{bmatrix},
    \end{equation*}
    where $\vec{G}_{\ell 1} \in \R^{q_\ell \x d_{\ell-1}}$ for some $q_\ell \in \mathbb{N}_{\geq 1}$, $\vec{G}_{\ell 2} \in \R^{(d_\ell - d_{\ell-1}) \x d_{\ell-1}}$, $\vec{I}_{d_{\ell-1}} \in \R^{d_{\ell-1} \x d_{\ell-1}}$ is the identity matrix, and $\varepsilon \in \R_{>0}$ is a constant. To minimize the number of free parameters required by our neural network Lyapunov candidate, we choose $q_\ell$ to be the minimum integer such that each entry in $\trans{\vec{G}_{\ell 1}}\vec{G}_{\ell 1} \in \R^{d_{\ell-1} \x d_{\ell-1}}$ is independent from the others. Since $\trans{\vec{G}_{\ell 1}}\vec{G}_{\ell 1}$ is symmetric, it has $\sum_{j=1}^{d_{\ell-1}} j = d_{\ell-1}(d_{\ell-1} + 1)/2$ free parameters, thereby requiring $q_\ell d_{\ell-1} \geq d_{\ell-1}(d_{\ell-1} + 1) / 2$ or $q_\ell \geq (d_{\ell-1} + 1) / 2$. For this, we choose $q_\ell = \left\lceil(d_{\ell-1} + 1)/2\right\rceil$.
\end{remark}

\end{document}